\documentclass[10pt]{article}

\usepackage{geometry}
\usepackage{amsfonts,epsfig,graphicx}
\usepackage{afterpage}
\usepackage{amsmath,amssymb,amsthm} 
\usepackage[T1]{fontenc} 
\usepackage{epsf} 
\usepackage{graphics} 
\usepackage{amsfonts,amsmath}
\usepackage[]{natbib} 
\usepackage{psfrag,xspace}
\usepackage{color,etoolbox}
\usepackage[]{algorithm2e}
\usepackage{algorithmic}

\usepackage{bbm}
\usepackage{sidecap}

\usepackage[normalem]{ulem}

\newtheorem{proposition}{Proposition}
\newtheorem{corollary}{Corollary}

\usepackage[utf8]{inputenc} 
\usepackage[T1]{fontenc}    
\usepackage[hidelinks]{hyperref}      
\usepackage{url}            
\usepackage{booktabs}       
\usepackage{amsfonts}       
\usepackage{nicefrac}       
\usepackage{microtype}      

\usepackage{microtype}
\usepackage{graphicx}
\usepackage{float}
\usepackage[export]{adjustbox}
\usepackage{subcaption}
\usepackage{booktabs}
\usepackage{xcolor}
\usepackage{enumerate}
\usepackage[shortlabels]{enumitem}
\usepackage{bm}
\usepackage{mathtools}

\newcommand{\Cset}{\mathcal{C}}

\newcommand{\Cest}{\widehat{C}}

\newcommand{\argmax}{\mathop{\mathrm{argmax}}}

\def\E{\mathbb{E}}
\def\P{\mathbb{P}}

\makeatletter
\long\def\@makecaption#1#2{
        \vskip 0.8ex
        \setbox\@tempboxa\hbox{\small {\bf #1:} #2}
        \parindent 1.5em  
        \dimen0=\hsize
        \advance\dimen0 by -3em
        \ifdim \wd\@tempboxa >\dimen0
                \hbox to \hsize{
                        \parindent 0em
                        \hfil 
                        \parbox{\dimen0}{\def\baselinestretch{0.96}\small
                                {\bf #1.} #2
                                } 
                        \hfil}
        \else \hbox to \hsize{\hfil \box\@tempboxa \hfil}
        \fi
        }
\makeatother

\newcommand{\prior}{\mathbb{P}_{pr}}
\newcommand{\perceived}{\Tilde{\ans}}
\newcommand{\perceivedRand}{\Tilde{Y}}
\newcommand{\predict}{\widehat{\ans}}
\newcommand{\truth}{\ans^*}
\newcommand{\ans}{y}
\newcommand{\data}{D}
\newcommand{\AIout}{f(M)}
\newcommand{\num}{N}
\newcommand{\pagree}[1]{P_{a_{#1}}}
\newcommand{\pagreer}[1]{P^r_{a_{#1}}}
\newcommand{\pagreew}[1]{P^w_{a_{#1}}}
\newcommand\given[1][]{\:#1\vert\:}

\newcommand{\tottime}[1]{T_{#1}}
\newcommand{\ho}{"Human only"}
\newcommand{\conf}{"Confidence-based time"}
\newcommand{\confExp}{"Confidence-based time with explanation"}
\newcommand{\random}{"Random time"}
\newcommand{\const}{"Constant time"}

\usepackage{etoolbox}

\title{Deciding Fast and Slow: The Role of Cognitive Biases in AI-assisted Decision-making}

\author{
Charvi Rastogi\footnote{Corresponding author. Email: \texttt{crastogi@cs.cmu.edu}} $^1$ , Yunfeng Zhang$^2$, Dennis Wei$^3$,\\ Kush R. Varshney$^3$, Amit Dhurandhar$^3$, Richard Tomsett$^4$ \vspace{2mm} \\ 
 $^1$Carnegie Mellon University, $^2$Twitter, $^3$IBM Research, $^4$Onfido
}

\date{}

\begin{document}
\maketitle

\begin{abstract}
       Several strands of research have aimed to bridge the gap between artificial intelligence (AI) and human decision-makers in AI-assisted decision-making, where humans are the consumers of AI model predictions and the ultimate decision-makers in high-stakes applications. However, people’s perception and understanding are often distorted by their cognitive biases, such as confirmation bias, anchoring bias, availability bias, to name a few. In this work, we use knowledge from the field of cognitive science to account for cognitive biases in the human-AI collaborative decision-making setting, and mitigate their negative effects on collaborative performance. To this end, we mathematically model cognitive biases and provide a general framework through which researchers and practitioners can understand the interplay between cognitive biases and human-AI accuracy. We then focus specifically on anchoring bias, a bias commonly encountered in human-AI collaboration. We implement a time-based de-anchoring strategy and conduct our first user experiment that validates its effectiveness in human-AI collaborative decision-making. With this result, we design a time allocation strategy for a resource-constrained setting that achieves optimal human-AI collaboration under some assumptions. We, then, conduct a second user experiment which shows that our time allocation strategy with explanation can effectively de-anchor the human and improve collaborative performance when the AI model has low confidence and is incorrect.
\end{abstract}

\section{Introduction}
\label{sec:introduction}
It should be a truth universally acknowledged that a human decision-maker in possession of an AI model must be in want of a collaborative partnership. Recently, we have seen a rapid increase in the deployment of machine learning (ML) models in decision-making systems, where the AI models serve as helpers to human experts in many high-stakes settings. Examples of such tasks can be found in healthcare, financial loans, criminal justice, job recruiting, and fraud monitoring. Specifically, judges use risk assessments to determine criminal sentences, banks use models to manage credit risk, and doctors use image-based ML predictions for diagnosis, to list a few.

The emergence of AI-assisted decision-making in society has raised questions about whether and when to rely on the AI model's decisions. These questions can be viewed as problems of communication between AI and humans, and research in interpretable, explainable, and trustworthy machine learning as efforts to improve aspects of this communication. However, a key component of human-AI communication that is often sidelined is the human decision-makers themselves. Humans' perception of the communication received from AI is at the core of this communication gap. Research in communication exemplifies the need to model receiver characteristics, thus implying the need to understand and account for human cognition in collaborative decision-making.

As a step towards studying human cognition in AI-assisted decision-making, our work focuses on the role of cognitive biases in this setting. Cognitive biases, introduced in the seminal work by~\citet{Tversky1124}, represent a systematic pattern of deviation from rationality in judgment wherein individuals create their own "subjective reality" from their perception of the input. An individual's perception of reality, not the objective input, may dictate their behavior in the world, thus, leading to distorted and inaccurate judgment. While cognitive biases and their effects on decision-making are well known and widely studied, we note that AI-assisted decision-making presents a new decision-making paradigm and it is important to study their role in this new paradigm, both analytically and empirically.

\begin{figure}[t]
\centering
    \includegraphics[ height=4.6cm]{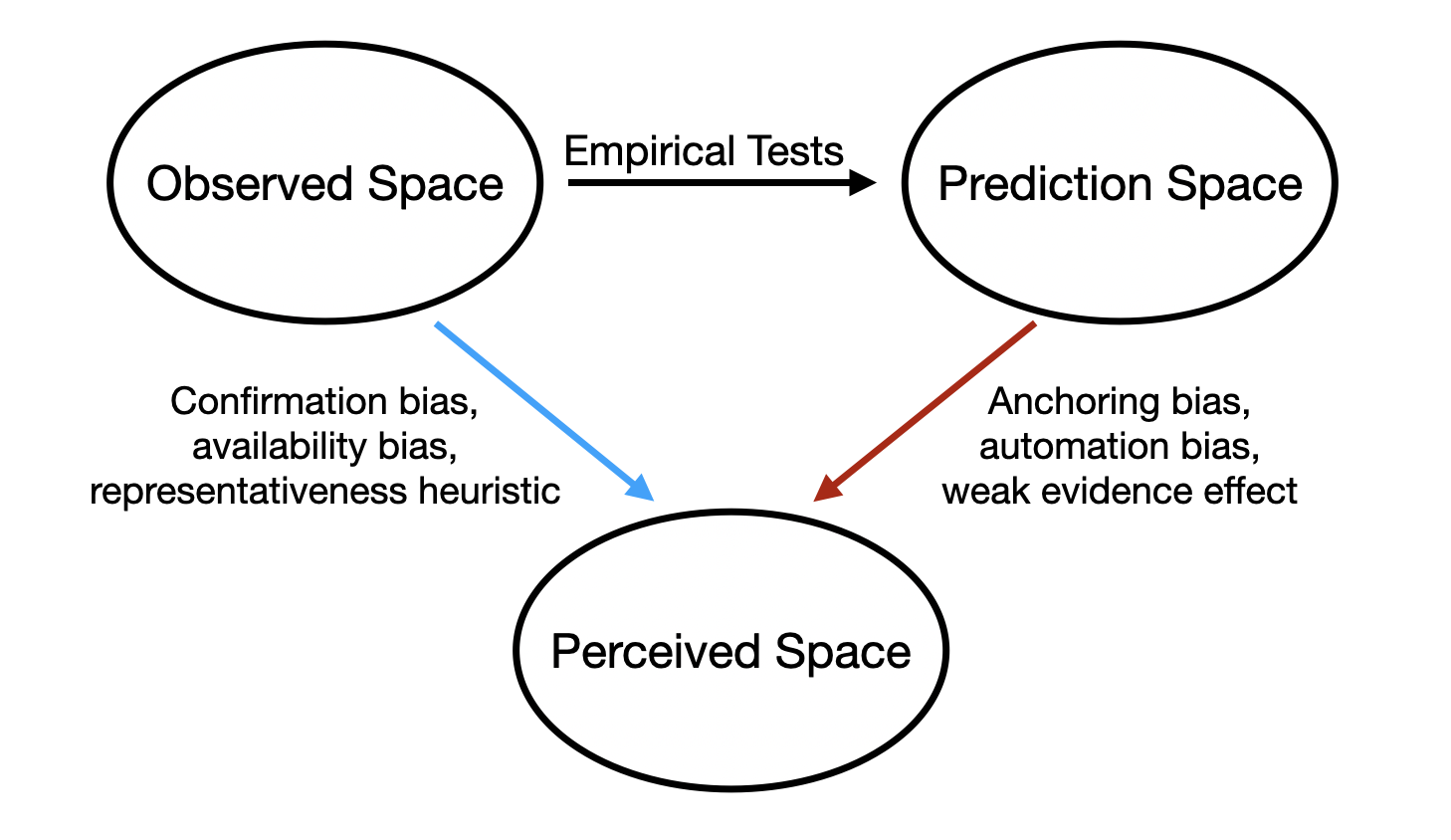}
    \caption{Three constituent spaces to capture different interactions in human-AI collaboration. The interactions of the perceived space, representing the human decision-maker, with the observed space and the prediction space may lead to cognitive biases. The definition of the different spaces is partially based on ideas of~\citet{yeom2018discriminative}.}
    \label{fig:explanatoryDiagram}
\end{figure}

Our first contribution is illustrated partially in Figure~\ref{fig:explanatoryDiagram}. In a collaborative decision-making setting, we define the perceived space to represent the human decision-maker. Here, we posit that there are two interactions that may lead to cognitive biases in the perceived space -- (1) interaction with the observed space which consists of the feature space and all the information the decision-maker has acquired about the task, (2) interaction with the prediction space representing the output generated by the AI model, which could consist of the AI decision, explanation, etc. In Figure~\ref{fig:explanatoryDiagram}, we associate cognitive biases that affect the decision-makers' priors and their perception of the data available with their interaction with the observed space. Confirmation bias, availability bias, the representativeness heuristic, and bias due to selective accessibility of the feature space are mapped to the observed space. On the other hand, anchoring bias and the weak evidence effect are mapped to the prediction space. Based on this categorization of biases, we provide a model for some of these biases using our biased Bayesian framework.

To focus our work in the remainder of the paper, we study and provide mitigating strategies for anchoring bias in AI-assisted decision-making, wherein the human decision-maker forms a skewed perception due to an anchor (AI decision) available to them, which limits the exploration of alternative hypotheses. Anchoring bias manifests through blind reliance on the anchor. While poorly calibrated reliance on AI has been studied previously from the lens of trust~\citep{zhang2020effect, tomsett2020rapid, okamura2020adaptive}, in this work we analyse reliance miscalibration  due to anchoring bias which has a different mechanism and, hence, different mitigating strategies.

\citet{Tversky1124} explained that anchoring bias manifests through the anchoring-and-adjustment heuristic wherein, when asked a question and presented with any anchor, people adjust away insufficiently from the anchor. Building on the notion of bounded rationality, previous work~\citep{Lieder2018anchoringBias} attributes the adjustment strategy to a resource-rational policy wherein the insufficient adjustment is a rational trade-off between accuracy and time. To test this in our setting, we conduct an experiment with human participants on Amazon Mechanical Turk to study whether allocating more resources --- in this case, time --- alleviates anchoring bias. This bias manifests through the rate of agreement with the AI prediction. Thus, by measuring the rate of agreement with the AI prediction in several carefully designed trials, we validate that time indeed is a useful resource that helps the decision-maker sufficiently adjust away from the anchor when needed. We note that the usefulness of time in remediating anchoring bias is an intuitive idea discussed in previous works~\cite{Tversky1124}, but one that has not been empirically validated to our knowledge. Thus, it is necessary to confirm this finding experimentally, especially in the AI-assisted setting.

As our first experiment confirms that more time helps reduce bounded-rational anchoring to the AI decision, one might suggest that giving more time to all decisions should yield better decision-making performance from the human-AI team. However, this solution does not utilize the benefits of high-quality AI decisions available in many cases. Moreover, it does not account for the limited availability of time. Thus, we formulate a novel resource (time) allocation problem that factors in the effects of anchoring bias and the variance in AI accuracy to maximize human-AI collaborative accuracy. We propose a time allocation policy and prove its optimality under some assumptions. We also conduct a second user experiment to evaluate human-AI team performance under this policy in comparison with several baseline policies. Our results show that while the overall performance of all policies considered is roughly the same, our policy helps the participants de-anchor from the AI prediction when the AI is incorrect and has low confidence.

The time allocation problem that we study is motivated by real-world AI-assisted decision-making settings. Adaptively determining the time allocated to a particular instance for the best possible judgement can be very useful in multiple applications. For example, consider a (procurement) fraud monitoring system deployed in multinational corporations, which analyzes and flags high-risk invoices \citep{procfraud}. Given the scale of these systems, which typically analyze tens of thousands of invoices from many different geographies daily, the number of invoices that may be flagged, even if a small fraction, can easily overwhelm the team of experts validating them. In such scenarios, spending a lot of time on each invoice is not admissible. An adaptive scheme that takes into account the biases of the human and the expected accuracy of the AI model is highly desirable to produce the most objective decisions. Our work is also applicable to the other aforementioned domains, such as in criminal proceedings where judges have to look over many different case documents and make quick decisions, often in under a minute.

In summary, we make the following contributions: 

\begin{itemize}
    \item We provide a biased Bayesian framework for modeling biased AI-assisted decision making. Based on the source of the cognitive biases, we situate some well-known cognitive biases within our framework.
    \item Focusing on anchoring bias in AI-assisted decision-making, we show with human participants that allocating more time to a decision reduces anchoring in this setting.
    \item We formulate a time allocation problem to maximize human-AI team accuracy that accounts for the anchoring-and-adjustment heuristic and the variance in AI accuracy.
    \item We propose a confidence-based allocation policy and identify conditions under which it achieves optimal team performance.
    \item Through a carefully designed human subject experiment, we evaluate the real-world effectiveness of the confidence-based time allocation policy, showing that  when confidence-based information is displayed, it helps humans de-anchor from incorrect and low-confidence AI predictions.

\end{itemize}
\section{Related work}
\label{sec:relatedWork}
Cognitive biases are an important factor in human decision-making~\citep{barnes1984cognitive, das1999cognitive, joyce2016decision}, and have been studied widely in decision-support systems research~\citep{arnott2006cognitive, zhang2015designing,solomon2014customization, gloria2019dss}. Cognitive biases also show up in many aspects of collaborative behaviours~\citep{silverman1992humancomp, janssen2020collaborative, rainer2010barriers}. More specifically, there exists decades-old research on cognitive biases~\citep{Tversky1124} such as confirmation bias~\citep{Nickerson1998ConfirmationBA, klayman1995varieties,oswald2004confirmation}, anchoring bias~\citep{furnham2011anchoring, epley-anchoring, epley2001putting}, automation bias~\cite{lee2004trust}, availability bias~\citep{tversky1973availability}, etc. 

Recently, as AI systems are increasingly embedded into high stakes human decisions, understanding human behavior, and reliance on technology have become critical, “Poor partnerships between people and automation will become increasingly costly and catastrophic”~\citep{lee2004trust}. This concern has sparked crucial research in several directions, such as human trust in algorithmic systems, interpretability, and explainability of machine learning models~\citep{arnold2019factsheets, zhang2020effect,tomsett2020rapid,  keng2018trust, doshi2017towards, lipton2018mythos, adadi2018peeking, preece2018asking}.

In parallel, research in AI-assisted decision-making has worked on improving the human-AI collaboration~\cite{lai2020chicago, lai2019explanations, bansal2020does, bansal2019updates, green2019principles, okamura2020adaptive}. These works experiment with several heuristic-driven AI explanation techniques that do not factor in all the characteristics of the human at the end of the decision-making pipeline. Specifically, the experimental results in~\citep{bansal2020does} show that explanations supporting the AI decision tend to exacerbate over-reliance on the AI decision. 
In contrast, citing a body of research in psychology, philosophy, and cognitive science, Miller~\cite{miller2019explanation} argues that the machine learning community should move away from imprecise, subjective notions of "good" explanations and instead focus on reasons and thought processes that people apply for explanation selection. In agreement with Miller, our work builds on literature in psychology on cognitive biases to inform modeling and effective de-biasing strategies. Our work provides a structured approach to addressing problems, like over-reliance on AI, from a cognitive science perspective. In addition, we adopt a two-step process, wherein we inform our subsequent de-biasing approach (Experiment 2) based on the results of our first experiment, thus, paving the pathway for experiment-driven human-oriented research in this setting.   

Work on cognitive biases in human-AI collaboration is still rare, however. Recently,  \citet{furnkranz2020cognitive} evaluated a selection of cognitive biases to test whether minimizing the complexity or length of a rule yields increased interpretability of machine learning models.~\citet{kliegr2018review} review twenty different cognitive biases that affect the interpretability and associated de-biasing techniques. Both these works~\citep{furnkranz2020cognitive, kliegr2018review} are specific to rule-based ML models.~\citet{baudel2020addressing} address complacency/authority bias in using algorithmic decision aids in business decision processes.~\citet{wang-chi-xaiframework} propose a conceptual framework for building explainable AI based on the literature on cognitive biases. Building on these works, our work provides novel mathematical models for the AI-assisted setting to identify the role of cognitive biases. Contemporaneously,~\citet{totrustbucinca2021} studies the use of cognitive forcing functions to reduce over-reliance in human-AI collaboration.

The second part of our work focuses on anchoring bias. The phenomenon of anchoring bias in AI-assisted setting has also been studied as a part of automation bias~\cite{lee2004trust} wherein the users display over-reliance on AI due to blind trust in automation. Previous experimental research has also shown that people do not calibrate their reliance on AI based on its accuracy~\citep{green2019principles}. Several studies suggest that people are unable to detect algorithmic errors~\citep{poursabzi2018manipulating}, are biased by irrelevant information~\citep{englich2006playing}, rely on algorithms that are described as having low accuracy, and trust algorithms that are described as accurate but present random predictions~\citep{springer2018dice}. These behavioural tendencies motivate a crucial research question --- how to account for these heuristics, often explained by cognitive biases such as anchoring bias. In this direction, our work is the first to empirically and analytically study a time-based de-biasing strategy to remediate anchoring bias in the AI-assisted setting. 

Lastly, the work by Park et al.~\citep{park2019slowAlgorithm} considers the approach of forcing decision-makers to spend more time deliberating their decision before the AI prediction is provided to them. In this work, we consider the setting where the AI prediction is provided to the decision-maker beforehand which may lead to anchoring bias. Moreover, we treat time as a limited resource and accordingly provide optimal allocation strategies.

\section{Problem setup and modeling}
\label{sec:setup}
 We consider a collaborative decision-making setup, consisting of a machine learning algorithm and a human decision-maker. First, we precisely describe our setup and document the associated notation. Following this, we provide a Bayesian model for various human cognitive biases induced by the human-AI collaborative process. 
 
Our focus in this paper is on the AI-assisted decision-making setup, wherein the objective of the human is to correctly classify the set of feature information available into one of two categories. Thus, we have a binary classification problem, where the true class is denoted by  $\truth\in \{0,1\}$. To make the decision/prediction, the human is presented with feature information, and we denote the complete set of features available pertaining to each sample by $\data$. In addition to the feature information, the human is also shown the output of the machine learning algorithm. Here, the AI output could consist of several parts, such as the prediction, denoted by $\predict\in \{0,1\}$, and the machine-generated explanation for its prediction. We express the complete AI output as a function of the machine learning model, denoted by $\AIout$.  Finally, we denote the decision made by the human decision-maker by $\perceived \in \{0,1\}$.
\\

We now describe the approach towards modeling the behavior of human decision-makers when assisted by machine learning algorithms.

\subsection{Bayesian decision-making}
\label{sec:bayesian}
Bayesian models for human cognition have become increasingly prominent across a broad spectrum of cognitive science~\citep{tenenbaum1999bayesian,griffiths2006optimal, chater2006cognition}. The Bayesian approach is thoroughly embedded within the framework of decision theory. Its basic tenets are that opinions should be expressed in terms of subjective or personal probabilities, and that the optimal revision of such opinions, in the light of relevant new information, should be accomplished via Bayes' theorem. 

First, consider a simpler setting, where the decision-maker uses the feature information available, $\data$, and makes a decision $\perceived\in\{0,1\}$. Let the decision variable be denoted by $\perceivedRand$. 
Based on literature in psychology and cognitive science~\citep{griffiths2006optimal, chater2006cognition}, we model a rational decision-maker as Bayes' optimal. That is, given a  prior on the likelihood of the prediction, $\prior(\perceivedRand)$ and the data likelihood distribution $\P(\data|\perceivedRand)$, the decision-maker picks the hypothesis/class with the higher posterior probability. Formally, the Bayes' theorem states that
\begin{align}
    \P(\perceivedRand=i|\data) = \dfrac{\P(\data|\perceivedRand=i)\prior(\perceivedRand=i)}{\sum_{j\in \{0,1\}}\P(\data|\perceivedRand=j)\prior(\perceivedRand=j)},
\end{align}
where $i\in \{0,1\}$ and the human decision is given by 
$\perceived = \argmax_{i\in\{0,1\}}\P(\perceivedRand=i|\data)$. Now, in our setting, in addition to the feature information available, the decision-maker takes into account the output of the machine learning algorithm, $\AIout$, which leads to following Bayes' relation 
\begin{align}
    \P(\perceivedRand|\data,\AIout)  \propto \P(\data,\AIout|\perceivedRand)\prior(\perceivedRand).
\end{align}
We assume that conditioned on the decision-maker's decision $\perceivedRand$, they perceive the feature information and the AI output independently, which gives 
\begin{align}\label{eq:naiveBayes}
    \P(\perceivedRand|\data,\AIout)  \propto \P(\data|\perceivedRand)\P(\AIout|\perceivedRand)\prior(\perceivedRand),
\end{align}
where $\P(\AIout|\perceivedRand)$ indicates the conditional probability of the AI output perceived by the decision-maker. The assumption in \eqref{eq:naiveBayes} is akin to a naive Bayes' assumption of conditional independence, but we only assume conditional independence between $\data$ and $\AIout$ and not between components within $\data$ or components within $\AIout$. This concludes our model for a rational decision-maker assisted by a machine learning model. 

In reality, the human decision-maker may behave differently from a fully rational agent due to their cognitive biases. 
In some studies~\citep{biased-bayesian-infer,payzan2011risk, payzan2012temp}, such deviations have  been  explained  by  introducing  exponential  biases (i.e. inverse temperature parameters) on Bayesian inference because these were found useful in expressing bias levels. We augment the modeling approach in~\citep{biased-bayesian-infer} to a human-AI collaborative setup. Herein we model the biased Bayesian estimation as 
\begin{align}
    \P(\perceivedRand|\data,\AIout)  \propto \P(\data|\perceivedRand)^\alpha \P(\AIout|\perceivedRand)^\beta \prior(\perceivedRand)^\gamma,
    \label{eq:exponent}
\end{align}
where $\alpha, \beta, \gamma$ are variables that represent the  biases in different factors in the Bayesian inference. 

Equation~\eqref{eq:exponent} allows us to understand and model several cognitive biases arising in AI-assisted decision making. To facilitate the following discussion, we take the ratio between \eqref{eq:exponent} evaluated for $\perceivedRand=1$ and \eqref{eq:exponent} for $\perceivedRand=0$:
\begin{equation}\label{eq:exponentRatio}
\frac{\P(\perceivedRand=1|\data,\AIout)}{\P(\perceivedRand=0|\data,\AIout)} = \left( \frac{\P(\data|\perceivedRand=1)}{\P(\data|\perceivedRand=0)} \right)^\alpha \left( \frac{\P(\AIout|\perceivedRand=1)}{\P(\AIout|\perceivedRand=0)} \right)^\beta \left( \frac{\prior(\perceivedRand=1)}{\prior(\perceivedRand=0)} \right)^\gamma.
\end{equation}
The human decision is thus $\perceivedRand=1$ if the ratio is greater than $1$ and $\perceivedRand=0$ otherwise. The final ratio is a product of the three ratios on the right-hand side raised to different powers. We can now state the following:
\begin{enumerate}
    \item  In anchoring bias, the weight put on AI prediction is high, i.e., $\beta >1$ and the corresponding ratio in \eqref{eq:exponentRatio} contributes more to the final ratio, whereas the weight on prior and data likelihood reduces.
    \item By contrast, in confirmation bias the weight on the prior is high, $\gamma > 1$, and the weight on the data and machine prediction reduces in comparison.
    \item Selective accessibility is a phenomena used to explain the mechanism of cognitive biases wherein the data that supports the decision-maker is selectively used as evidence, while the rest of the data is not considered. This distorts the data likelihood factor in \eqref{eq:exponent}. The direction of distortion $\alpha > 1$ or $\alpha < -1$ depends on the cognitive bias driven decision. 
    \item The weak evidence effect~\citep{fernbach2011weakEvidence}  suggests that when presented with weak evidence for a prediction, the decision-maker would tend to choose the opposite prediction. This effect is modeled with $\beta < -1$. 
\end{enumerate}
 
\noindent To focus our approach, we consider a particular cognitive bias --- anchoring bias, which is specific to the nature of human-AI collaboration and has been an issue in previous works~\citep{lai2019explanations, springer2018dice, bansal2020does}. In the next section, we summarise the findings about anchoring bias in the literature, explain proposed de-biasing technique and conduct an experiment to validate the technique.

\section{Anchoring bias}
\label{sec:anchoringBias}
AI-assisted decision-making tasks are prone to anchoring bias, where the human decision-maker is irrationally anchored to the AI-generated decision. 
The anchoring-and-adjustment heuristic, introduced by Tversky and Kahneman in~\cite{Tversky1124} and studied in~\citep{epley-anchoring, Lieder2018anchoringBias} suggests that after being anchored, humans tend to adjust insufficiently because adjustments are effortful and tend to stop once a plausible estimate is reached. Lieder et al.~\citep{Lieder2018anchoringBias} proposed the resource rational model of anchoring-and-adjustment which explains that the insufficient adjustment can be understood as a rational trade-off between time and accuracy.  
This is a consequence of the bounded rationality of humans~\citep{simon1956rational, simon1972theories}, which entails satisficing, that is, accepting  sub-optimal solutions that are  good enough, rather than optimizing solely for accuracy. Through user studies, Epley et al.~\citep{epley-anchoring} argue that cognitive load and time pressure are contributing factors behind insufficient adjustments. 

Informed by the above works viewing anchoring bias as a problem of insufficient adjustment due to limited resources, we aim to mitigate the effect of anchoring bias in AI-assisted decision-making, using time as a resource. We use the term de-anchoring to denote the rational process of adjusting away from the anchor. With this goal in mind, we conducted two user studies on Amazon Mechanical Turk. Through the first user study (Experiment 1), we aim to understand the effect of different time allocations on anchoring bias and de-anchoring in an AI-assisted decision-making task. In Experiment 2, we use the knowledge obtained about the effect of time in Experiment 1 to design a time allocation strategy and test it on the experiment participants.  

We now describe Experiment 1 in detail.

\subsection{Experiment 1}  
\label{sec:prelimStudy}
In this study, we asked the participants to complete an AI-assisted binary prediction task consisting of a number of trials. Our aim is to learn the effect of allocating different amounts of time to different trials on participants with anchoring bias. 

To quantify anchoring bias and thereby the insufficiency of adjustments, we use the probability $\P(\perceived = \predict)$ that the human decision-maker agrees with the AI prediction $\predict$, which is easily measured. This measure can be motivated from the biased Bayesian model in \eqref{eq:exponent}. In the experiments, the model output $f(M)$ consists of only a predicted label $\predict$. In this case, \eqref{eq:exponent} becomes 
\begin{align}
    &\P(\perceivedRand=y|\data,\AIout) \propto \P(\data|\perceivedRand=y)^\alpha \P(\widehat{Y}=\predict|\perceivedRand=y)^\beta \prior(\perceivedRand=y)^\gamma.
    \label{eq:exponentYhat}
\end{align}
 Let us make the reasonable assumption that the decision-maker's decision $\perceivedRand$ positively correlates with the AI prediction $\widehat{Y}$, specifically that the ML model's probability $\P(\widehat{Y}=\predict|\perceivedRand=y)$ is larger when $y = \predict$ than when $y \neq \predict$. Then as the exponent $\beta$ increases, i.e., as anchoring bias strengthens, the likelihood that $y = \predict$ maximizes \eqref{eq:exponentYhat} and becomes the human decision $\perceived$ also increases. In the limit $\beta \to \infty$, we have agreement $\perceived = \predict$ with probability $1$. Conversely, for $\beta = 1$, the two other factors in \eqref{eq:exponentYhat} are weighed appropriately and the probability of agreement assumes a natural baseline value. We conclude that the probability of agreement is a measure of anchoring bias. It is also important to ensure that this measure is based on tasks where the human has reason to choose a different prediction. 

Thus, given the above relationship between anchoring bias and agreement probability (equivalently disagreement probability), we tested the following hypothesis to determine whether time is a useful resource in mitigating anchoring bias:
\begin{itemize}
    \item Hypothesis 1 (H1): Increasing the time allocated to a task alleviates anchoring bias, yielding a higher likelihood of sufficient adjustment away from the AI-generated decision when the decision-maker has the knowledge required to provide a different prediction.
\end{itemize}

\paragraph{\textbf{Participants.}} We recruited 47 participants from Amazon Mechanical Turk for Experiment 1, limiting the pool to subjects from within the United States with a prior task approval rating of at least 98\% and a minimum of 100 approved tasks. 10 participants were between Age 18 and 29, 26 between Age 30 and 39, 6 between Age 40 and 49, and 5 over Age 50. The average completion time for this user study was 27 minutes, and each participant received compensation of \$4.5 (roughly equals an hourly wage of \$10). The participants received a base pay of \$3.5 and a bonus of \$1 (to incentivize accuracy).   

\paragraph{\textbf{Task and AI model.}} We designed a performance prediction task wherein a participant was asked to predict whether a student would pass or fail a class, based on the student's characteristics, past performance, and some demographic information. The dataset for this task was obtained from the UCI Machine Learning Repository, published as the Student Performance Dataset~\citep{cortez2008using}. This dataset contains 1044 instances of students' class performances in 2 subjects (Mathematics and Portuguese), each described by 33 features. To prepare the dataset, we binarized the target labels (`pass', `fail'), split the dataset into training and test sets (70/30 split). To create our AI, we trained a logistic regression model on the standardized set of features from the training dataset. Based on the feature importance (logistic regression coefficients) assigned to each feature in the dataset, we retained the top 10 features for the experiments. These included --- mother's and father's education, mother's and father's jobs, hours spent studying weekly, interest in higher education, hours spent going out with friends weekly, number of absences in the school year, enrolment in extra educational support, and number of past failures in the class.

\paragraph{\textbf{Study procedure}}
\label{sec:studyProcedure1}Since we are interested in studying decision-makers' behavior when humans have prior knowledge and experience in the prediction task, we first trained our participants before collecting their decision data for analysis. The training section consists of 15 trials where the participant is first asked to provide their prediction based on the student data and is then shown the correct answer after attempting the task. These trials are the same for all participants and are sampled from the training set such that the predicted probability (of the predicted class) estimated by the AI model is distributed uniformly over the intervals $[0.5,0.6], (0.6,0.7], \cdots, (0.9 ,1]$. Taking predicted probability as a proxy for difficulty, this ensures that all levels of difficulty are represented in the task. To help accelerate participants' learning, we showed bar charts that display the distributions of the outcome across the feature values of each feature. These bar charts were not provided in the testing section to ensure stable performance throughout the testing section and to emulate a real-world setting.

To induce anchoring bias, the participant was informed at the start of the training section that the AI model was $85\%$ accurate (we carefully chose the training trials to ensure that the AI was indeed $85\%$ accurate over these trials), while the model's actual accuracy is $70.8\%$ over the entire training set and $66.5\%$ over the test set. Since our goal is to induce anchoring bias and the training time is short, we stated a high AI accuracy. Moreover, this disparity between stated accuracy ($85\%$) and true accuracy ($70.8\%$) is realistic if there is a distribution shift between the training and the test set, which would imply that the humans' trust in AI is misplaced. In addition to stating AI accuracy at the beginning, we informed the participants about the AI prediction for each training trial after they have attempted it so that they can learn about AI's performance first-hand.

The training section is followed by the testing section which consists of 36 trials sampled from the test set and was kept the same (in the same order) for all participants. In this section, the participants were asked to make a decision based on both the student data and AI prediction. They were also asked to describe their confidence level in their prediction as low, medium, or high.

To measure the de-anchoring effect of time, we included some trials where the AI is incorrect but the participants have the requisite knowledge to adjust away from the incorrect answer. That is, we included trials where the participants' accuracy would be lower when they are anchored to the AI prediction than when they are not. We call these trials --- \emph{probe trials}, which help us probe the effect of time on de-anchoring. On the flip side, we could not include too many of these trials because participants may lose their trust in the AI if exposed to many apparently incorrect AI decisions. To achieve this balance, we sampled 8 trials of medium difficulty where the AI prediction is accurate (predicted probability ranging from $0.6$ to $0.8$) and flip the AI prediction shown to the participants. The remaining trials, termed \emph{unmodified trials} are sampled randomly from the test set while maintaining a uniform distribution over the AI predicted probability (of the predicted class). Here, again, we use the predicted probability as a proxy for the difficulty of the task, as evaluated by the machine learning model. We note that the accuracy of the AI predictions \emph{shown} to the participants is $58.3\%$ which is far lower than the $85\%$ accuracy shown in the training section.

\paragraph{\textbf{Time allocation.}}
\label{sec:timeAllocation} To investigate the effect of time on the anchoring-and-adjustment heuristic in AI-assisted decision-making, we divide the testing section into four blocks for each participant based on the time allocation per trial. To select the allocated time intervals, we first conducted a shorter version of the same study to learn the amount of time needed to solve the student performance prediction task, which suggested that the time intervals of $10$s, $15$s, $20$s and $25$s captured the range from necessary to sufficient. Now, with these four time intervals, we divided the testing section into four blocks of 9 trials each, where the time allocated per trial in each block followed the sequence $[t_1, t_2, t_3, t_4]$ and for each participant this sequence was a random permutation of $[10,15,20,25]$. The participants were not allowed to move to the next trial till the allocated time ran out. Furthermore, each block was comprised of 2 probe trials and 7 unmodified trials, randomly ordered. Recall that each participant was provided the same set of trials in the same order.

 Now, with the controlled randomization of the time allocation, independent of the participant, their performance, and the sequence of the tasks, we are able to identify the effect of time on de-anchoring. It is possible that a participant that disagrees with the AI prediction often in the first half, is not anchored to the AI prediction in the latter half. Our study design allows us to average out such participant-specific effects, through the randomization of time allocation interval sequences across participants.


\begin{figure}%
    \centering
    \subfloat[]{\includegraphics[width=6.5cm]{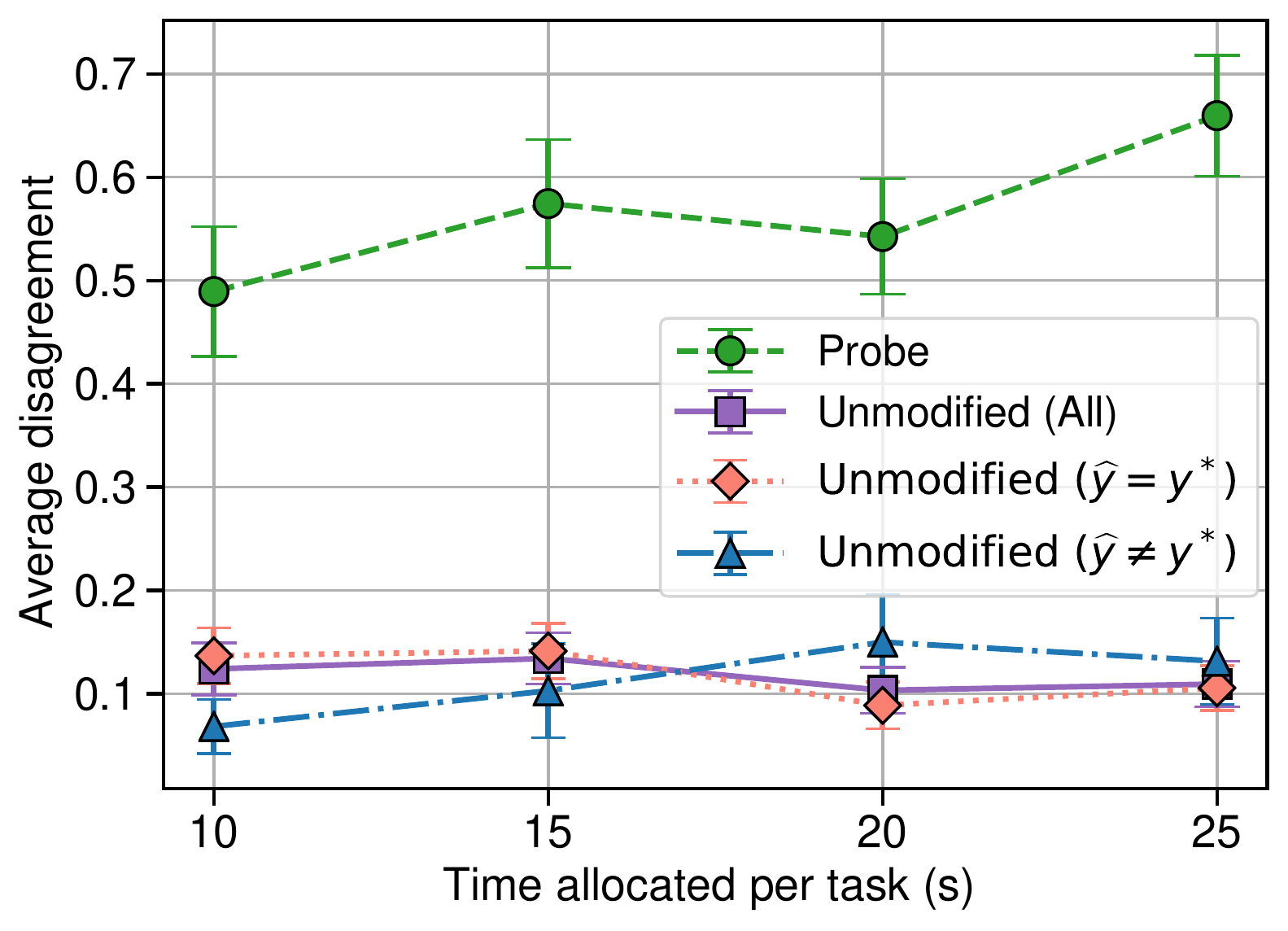} } 
    \subfloat[]{\includegraphics[width=5.1cm]{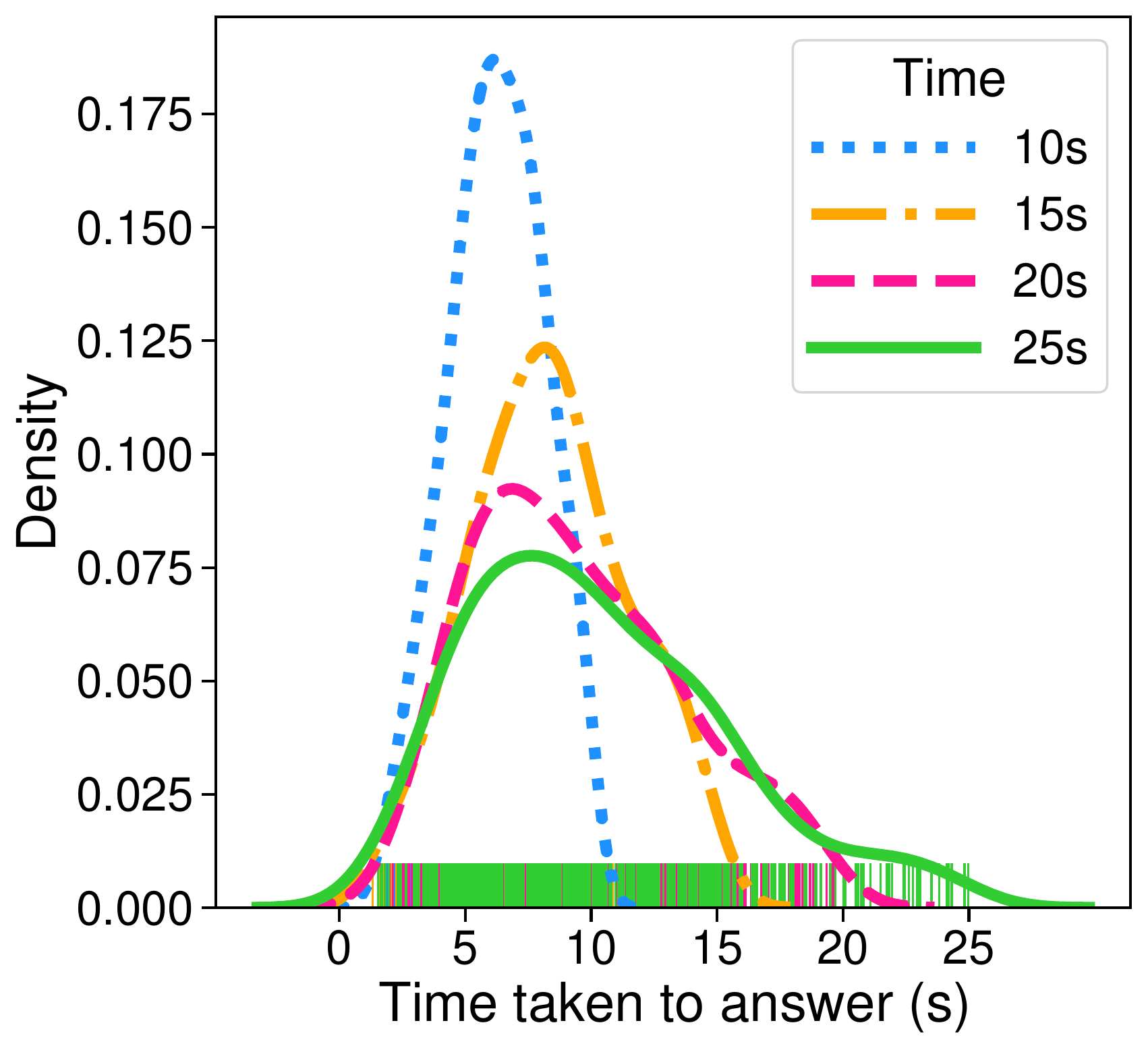} }
    \caption{Results of experiment 1. \textbf{(a)} Average disagreement with the AI prediction for different time allocations in experiment 1. \textbf{(b)} Distribution of time taken by the participants to provide their answer under the four different time conditions, \{10,15,20,25\} seconds in Experiment 1. For illustration purposes, we use kernel density estimation to estimate the probability density function shown. The actual answering time lies between 0 and 25 seconds.}
    \label{fig:timeUse}
\end{figure}

The main results of Experiment 1 are illustrated in Figure~\ref{fig:timeUse}(a). We see that the probe trials served their intended purpose, since the average disagreement is much higher for probe trials compared to unmodified trials for all time allocations. This suggests that the participants had learned to make accurate predictions for this task, otherwise they would not be able to detect the AI's errors in the probe trials, more so in the $10$-second condition. 
We also observe that the likelihood of disagreement for unmodified trials is low (close to 0.1) for all time allocations. This suggests that the participants' knowledge level in this task is roughly similar to or less than that of the AI since the participants are unable to offer any extra knowledge in the unmodified trials.

\paragraph{\textbf{Anchoring-and-adjustment.}}  The results on the probe trials in Figure~\ref{fig:timeUse}(a) suggest that the participants' likelihood of sufficiently adjusting away from the incorrect AI prediction increased as the time allocated increased. This strengthens the argument that the anchoring-and-adjustment heuristic is a resource-rational trade-off between time and accuracy~\citep{Lieder2018anchoringBias}. 
Specifically, we observe that the average disagreement percentage in probe trials increased from $48\%$ in the $10$-second condition to $67\%$ in the $25$-second condition. We used the bootstrap method with 5000 re-samples to estimate the coefficient of a linear regression fit on average disagreement vs. time allocated for probe trials. This resulted in a significantly positive coefficient of $0.01$ (bootstrap $95\%$ confidence interval $[0.001,0.018])$.  This result is consistent with our Hypothesis 1 (H1) that increasing time for decision tasks alleviates anchoring bias. We note that the coefficient is small in value because the scales of the independent and dependent variables of the regression (time and average disagreement) have not been adjusted for the regression, so the coefficient of $0.01$ yields a $~0.15$ increase in average disagreement between the 10s and the 25s time condition.

\paragraph{\textbf{Time adherence.}} Figure~\ref{fig:timeUse}(b) suggests that the participants adhere reasonably to the four different time conditions used. We note that this time reflects the maximum time taken to click the radio button (in case of multiple clicks), but the participants may have spent more time thinking over their decision. In the survey at the end of the study, we asked the participants how often they used the entire time available to them in the trials, and obtained the following distribution of answers --- Frequently 15, Occasionally 24, Rarely 6, Never 2.


\section{Optimal resource allocation in human-AI collaboration}

In Section~\ref{sec:anchoringBias}, we see that time is a useful resource for de-anchoring the decision-maker. More generally, there are many works that study de-biasing techniques to address the negative effects of cognitive biases. These de-biasing techniques require resources such as time, computation and explanation strategies. Thus, in this section, we model the problem of mitigating the effect of cognitive biases in the AI-assisted decision-making setting as a resource allocation problem, where our aim is to efficiently use the resources available and improve human-AI collaboration accuracy.
\subsection{Resource allocation problem}
\label{sec:resourceAllocationProblem}
 From Experiment 1 in Section~\ref{sec:prelimStudy}, we learnt that given more time, the decision-maker is more likely to adjust away from the anchor (AI decision) if the decision-maker has reason to believe that the correct answer is different. This is shown by change in their probability of agreement with the AI prediction, denoted by $\pagree{} = \P(\perceived=\predict)$.  The results of Experiment 1 indicate that, ideally, decision-makers should be provided with ample time for each decision. However, in practice, given a finite resource budget $\tottime{}$, we also have the constraint $\sum_{i=1}^\num \tottime{i} = \tottime{}$. Thus, we formulate a resource allocation problem that captures the trade-off between time and accuracy. More generally, this problem suggests a framework for optimizing \emph{human-AI} team performance using constrained resources to de-bias the decision-maker.

In our setup, the human decision-maker has to provide a final decision $\perceived_i$ for $\num$ total trials with AI assistance, specifically in the form of a predicted label $\predict_i$. The objective is to maximize the average accuracy over the trials, denoted by $R$, of the human-AI collaboration: $\E[R] = \dfrac{1}{\num} \sum_{i=1}^\num \E[R_i]$,
where $R_i$ is an indicator of human-AI correctness in trial $i$. 

We first relate collaborative accuracy $\E[R]$ to the anchoring-and-adjustment heuristic. Intuitively, if we know the AI to be incorrect in a given trial, we should aim to facilitate adjustment away from the anchor as much as possible, whereas if AI is known to be correct, then anchoring bias is actually beneficial. Based on this intuition, $E[R_i]$ can be rewritten by conditioning on AI correctness/incorrectness as follows: 
\begin{align}
    \E[R_i] &= 
    \underbrace{\P(\perceived_i = \predict_i\given \predict_i = \truth_i)}_{\pagreer{i}} \P(\predict_i = \truth_i)  + \bigl(1-\underbrace{\P(\perceived_i = \predict_i\given \predict_i \neq \truth_i)}_{\pagreew{i}}\bigr)\left(1-\P(\predict_i=\truth_i)\right). 
    \label{eq:expReward}
\end{align}
We see therefore that human-AI correctness depends on the probability of agreement $\pagreer{i}$ conditioned on AI being correct and the probability of agreement $\pagreew{i}$ conditioned on AI being incorrect. Recalling from Section~\ref{sec:anchoringBias} the link established between agreement probability and anchoring bias, \eqref{eq:expReward} shows the effect of anchoring bias on human-AI accuracy. Specifically in the case of \eqref{eq:expReward}, the effect is through the two conditional agreement probabilities $\pagreer{i}$ and $\pagreew{i}$.

We consider time allocation strategies to modify agreement probabilities and thus improve collaborative accuracy, based on the relationship established in Experiment 1. 

We denote the time 
used in trial $i$ as $\tottime{i}$, which impacts correctness $R_i$ \eqref{eq:expReward} as follows: 
\begin{equation}\label{eq:RiTi}
    \E[R_i \given \tottime{i}] = \pagreer{i}(\tottime{i}) \P(\predict_i = \truth_i) + \pagreew{i}(\tottime{i}) \left(1 - \P(\predict_i = \truth_i)\right).
\end{equation}
The allocation of time affects only the human decision-maker, making the agreement probabilities functions of $\tottime{i}$, whereas the probability of AI correctness is unaffected. The resource allocation problem would then be to maximize the average of \eqref{eq:RiTi} over trials subject to the budget constraint $\sum_{i=1}^n T_i=T$. 

\begin{figure}%
    \centering
    \includegraphics[width=6cm]{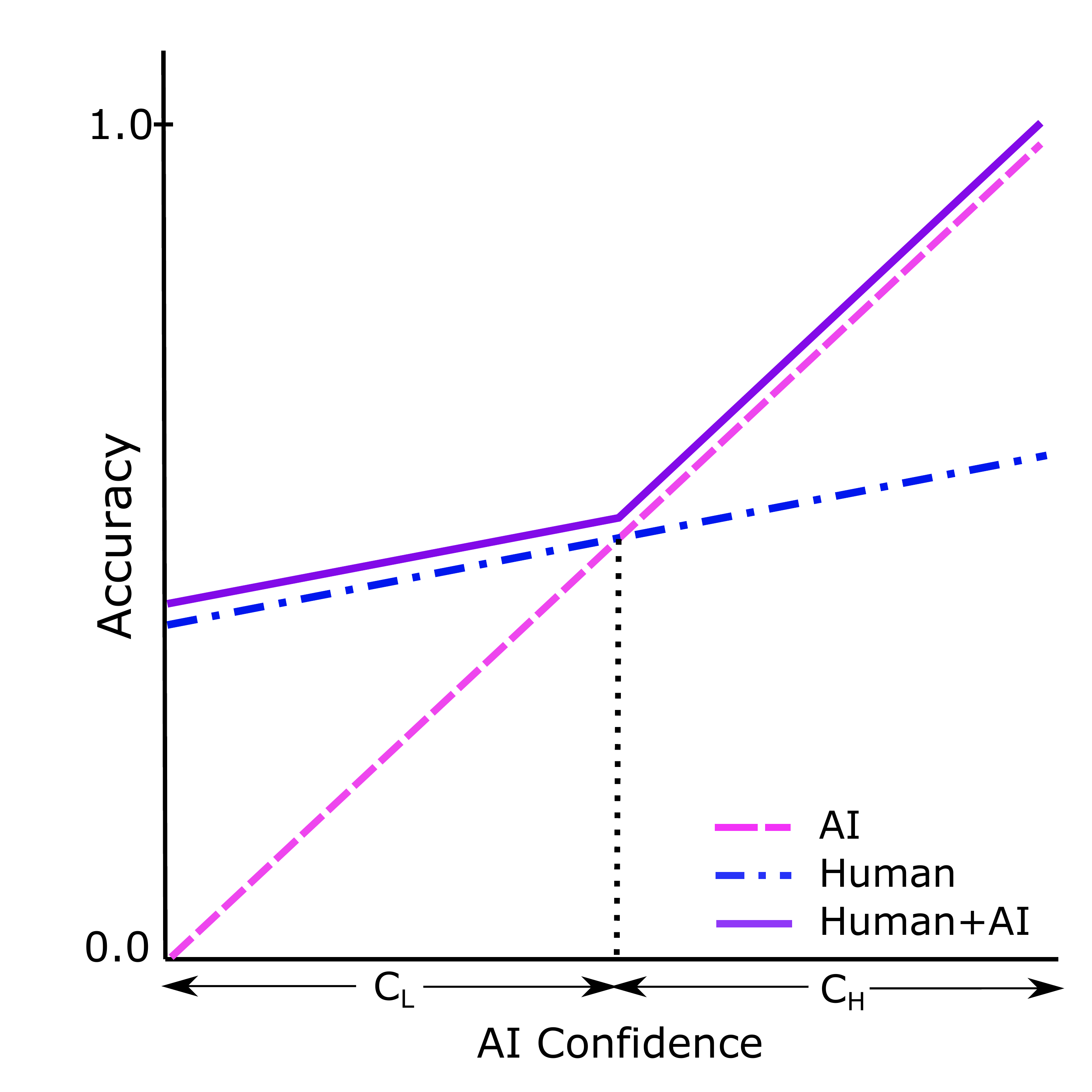} 
    \caption{An ideal case for human-AI collaboration, where (1) we correctly identify the set of tasks with low and high AI confidence, (2) the AI accuracy is perfectly correlated with its confidence, (3) human accuracy is higher than AI in the low confidence region, $\Cset_L$, and lower than AI in the high confidence region $\Cset_H$. }%
    \label{fig:idealhumanAI}
\end{figure}

The challenge with formulation \eqref{eq:RiTi} is that it requires
identifying the true probability of AI correctness, 
which is a non-trivial task~\citep{guo2017calibration}. Instead, we operate under the more realistic assumption that the AI model can estimate its probability of correctness from the class probabilities that it predicts (as provided for example by a 
logistic regression model). We refer to this estimate as \emph{AI confidence} and denote it as $\Cest_i$. We may then consider a decomposition of human-AI correctness as in \eqref{eq:expReward}, \eqref{eq:RiTi} but conditioned on $\Cest_i$. In keeping with the two cases in \eqref{eq:expReward}, \eqref{eq:RiTi} and to simplify the allocation strategy, we binarize $\Cest_i$ into two intervals, low confidence $\Cest_i \in \Cset_L$, and high confidence $\Cest_i \in \Cset_H$. The time allocated is then $T_i(\Cest_i) = t_L$ for $\Cest_i \in \Cset_L$ and $T_i(\Cest_i) = t_H$ for $\Cest_i \in \Cset_H$. Thus we have 
\begin{align}\label{eq:RiCi}
    \E[R_i] &= \P(\Cest_i \in \Cset_L) \E[R_i \given \Cest_i \in \Cset_L, \tottime{i} = t_L] + \P(\Cest_i \in \Cset_H) \E[R_i \given \Cest_i \in \Cset_H, \tottime{i} = t_H].
\end{align}
The quantities $\E[R_i \given \Cest_i \in \Cset, \tottime{i}]$, $\Cset = \Cset_L, \Cset_H$, are not pure agreement probabilities as in \eqref{eq:RiTi} because the low/high-confidence events $\Cest_i \in \Cset_L$, $\Cest_i \in \Cset_H$ generally differ from the correctness/incorrectness events $\predict_i = \truth_i$, $\predict_i \neq \truth_i$. Nevertheless, since we expect these events to be correlated, $\E[R_i \given \Cest_i \in \Cset, \tottime{i}]$ is related to the agreement probabilities in \eqref{eq:RiTi}.

Figure~\ref{fig:idealhumanAI} presents an ideal scenario that one hopes to attain in \eqref{eq:RiCi}. In presence of anchoring bias, our aim is to achieve the human-AI team accuracy shown. This approach capitalises on human expertise where AI accuracy is low. Specifically, by giving human decision-makers more time, we encourage them to rely on their own knowledge (de-anchor from the AI prediction) when the AI is less confident, $\Cest_i \in \Cset_L$. Usage of more time in low AI confidence tasks, implies less time in tasks where AI is more confident, $\Cest_i \in \Cset_H$, where anchoring bias has lower negative effects and is even beneficial. Thus, this two-level AI confidence based time allocation policy allows us to mitigate the negative effects of anchoring bias and achieve the ``best of both worlds'', as illustrated in Figure~\ref{fig:idealhumanAI}. 

We now formally write the assumption under which the optimal time allocation policy is straightforward to see. 
\par \noindent
\textbf{Assumption 1.} For any $t_1, t_2 \in \mathbb{R}^{+}$, if $t_1 < t_2$, then 
\begin{align}
\begin{split}
     & \E[R_i \given \Cest_i \in \Cset_L, \tottime{i} = t_1] \leq \E[R_i \given \Cest_i \in \Cset_L, \tottime{i} = t_2],\,\,\textnormal{and}\\
    & \E[R_i \given \Cest_i \in \Cset_H, \tottime{i} = t_1] \geq \E[R_i \given \Cest_i \in \Cset_H, \tottime{i} = t_2]. 
    \label{eq:assumption1}
\end{split}
\end{align}
We provide explanation for Assumption 1 \eqref{eq:assumption1} in Appendix A. Under Assumption~1, the optimal strategy is to maximize time for low-AI-confidence trials and minimize time for high-confidence trials, as is stated formally below.
\begin{proposition}
Consider the AI-assisted decision-making setup discussed in this work with $\num$ trials where the total time available is $T$. Suppose Assumption 1, stated in~\eqref{eq:assumption1}, holds true for human-AI accuracy. Then, the optimal confidence-based allocation is as follows, 
\begin{align}
  \tottime{i} =
\left\{
	\begin{array}{ll}
		t_H =  t_{\min} & \mbox{if } \;\;\Cest_i \in \Cset_H \\
		t_L =  t_{\max} & \mbox{if } \;\;\Cest_i \in \Cset_L,
	\end{array}
    \label{eq:confBased}
\right.
\end{align}
 where $t_{\min}$ is the minimum allowable time, and $ t_{\max}$ is the corresponding maximum time such that $t_{\max} \P(\Cest_i \in \Cset_L) + t_{\min} \P(\Cest_i \in \Cset_H) = \frac{T}{N}$. 
\label{prop:optimal}
\end{proposition}
Proposition~\ref{prop:optimal} gives us the optimal time allocation strategy by efficiently allocating more time for adjusting away from the anchor in the tasks that yield a lower probability of accuracy if the human is anchored to the AI predictions. We note that, although in an ideal scenario as shown in Figure~\ref{fig:idealhumanAI}, we should set $t_{\min}=0$, in real world implementation $\Cest_i$ is an approximation of the true confidence, and hence, it would be helpful to have human oversight with $t_{\min} > 0$ in case the AI confidence is poorly calibrated.

To further understand the optimality of the confidence-based time allocation strategy, we compare it with two baseline strategies that obey the same resource constraint, namely, Constant time and Random time strategies, defined as --\begin{itemize}
    \item Constant time: For all $i$, $\tottime{i} = \frac{\tottime{}}{\num}$.
    \item Random time : Out of the $\num$ trials, $\num\times\P(\Cest_i \in \Cset_L)$ trials are selected randomly and allocated time $t_L$. The remaining trials are allocated time $t_H$.
\end{itemize}
Constant time is the most natural baseline allocation, while Random time assigns the same values $t_L$ and $t_H$ as the confidence-based policy but does so at random. Both are evaluated in the experiment described in Section~\ref{sec:mainStudy}.

\subsection{Experiment 2: Dynamic time allocation for human-AI collaboration}
\label{sec:mainStudy}
In this experiment, we implement our confidence-based time allocation strategy for human-AI collaboration in a user study deployed on Amazon Mechanical Turk. Based on the results of Experiment 1 shown in Figure~\ref{fig:timeUse}(a), we assign $t_L = 25s$ and $t_H = 10s$. 

In addition, we conjecture that giving the decision-maker the reasoning behind the time allocation, that is, informing them about AI confidence and then allocating time accordingly, would help improve the collaboration further. This conjecture is supported by findings in~\citep{Chambon2020choosing}, where the authors observe that choice tips the balance of learning: for the same action and outcome, the brain learns differently and more quickly from free choices than forced ones. Thus, providing valid reasons for time allocation would help the decision-maker make an active choice and hence, learn to collaborate with AI better.

In this experiment, we test the following hypotheses.
\begin{itemize}
    \item H2: Anchoring bias has a negative effect on human-AI collaborative decision-making accuracy when AI is incorrect.
    \item H3: If the human decision-maker has complementary knowledge then allocating more time can help them sufficiently adjust away from the AI prediction.
    \item H4 : Confidence-based time allocation yields better performance than Human alone and AI alone. 
    \item  H5: Confidence-based time allocation yields better human-AI team performance than constant time and random time allocations.
    \item H6: Confidence-based time allocation with explanation yields better human-AI team performance than the other conditions.
\end{itemize}
We now describe the different components of Experiment 2 in detail. 

 \paragraph{\textbf{Participants.}}
In this study, 479 participants were recruited in the same manner as described in Section~\ref{sec:prelimStudy}. 83 participants were between ages 18 and 29, 209 between ages 30 and 39, 117 between ages 40 and 49, and 70 over age 50. The average completion time for this user study was 30 minutes, and participants received compensation of \$5.125 on average (roughly equals an hourly wage of \$10.25). The participants received an average base pay of \$4.125 and bonus of \$1 (to incentivize accuracy).   

\paragraph{\textbf{Task and AI model.} } The binary prediction task in this study is the same as the student performance prediction task used before. In this experiment, our goal is to induce optimal human-AI collaboration under the assumptions illustrated in Figure~\ref{fig:idealhumanAI}. In real-world human-AI collaborations, it is not uncommon for the decision-maker to have some domain expertise or complementary knowledge that the AI does not, especially in fields where there is not enough data such as social policy-making and design. To emulate this situation where the participants have complementary knowledge, we reduced the information available to the AI, given the unavailability of human experts and the limited training time in our experiment. We train the assisting AI model over 7 features, while the participants have access to 3 more features, namely, hours spent studying weekly, hours spent going out with friends weekly, and enrollment in extra educational support. These 3 features were the second to fourth most important ones as deemed by a full model.

To implement the confidence-based time allocation strategy, we had to identify trials belonging to classes $\Cset_L$ and $ \Cset_H$. Ideally, for this we require a machine learning algorithm that can calibrate its confidence correctly. As discussed in Section~\ref{sec:resourceAllocationProblem}, we use the AI's predicted probability $\Cest_i$ (termed as AI confidence) and choose the threshold for $\Cset_H$ as $\Cest_i \geq 0.75$. This study has $40$ questions in the testing section, from which $20$ belong to $\Cset_L$ and $20$ belong to $\Cset_H$.

\paragraph{\textbf{Study procedure.}} As in Experiment 1, this user study has two sections, the training section and the testing section. The training section is exactly the same as before where the participants are trained over $15$ examples selected from the training dataset. To induce anchoring bias, as in Experiment 1, we reinforce that the AI predictions are $85\%$ accurate in the training section. 

The testing section has 40 trials, which are sampled randomly from the test set such that the associated predicted probability values (of the predicted class) estimated by the machine learning algorithm are distributed uniformly. While the set of trials in the testing section is fixed for all participants, the order they were presented in was varied randomly. 

To test hypotheses H2, H3, H4, H5 and H6, we randomly assigned each participant to one of five groups:  
\begin{enumerate}
    \item \textbf{Human only}: In this group, the participants were asked to provide their prediction without the help of the AI prediction. The time allocation for each trial in the testing section is fixed at $25$ seconds. This time is judged to be sufficient for humans to make a prediction on their own, based on the results of Experiment 1 (for example the time usage distributions in Figure~\ref{fig:timeUse}). 
    \item \textbf{Constant time}: In this group, the participants were asked to provide their prediction with the help of the AI prediction. The time allocation for each trial in the testing section is fixed as $\frac{t_L + t_H}{2} = 17.5$ seconds. We rounded this to $18$ seconds when reporting it to the participants. 
    \item \textbf{Random time}: This group has all factors the same as the constant time group except for the time allocation. For each participant, the time allocation for each trial is chosen uniformly at random from the set $\{10, 25\}$ such that the average time allocated per trial is $17.5$ seconds.
    \item \textbf{Confidence-based time}: This is our treatment group, where we assign time according to confidence-based time allocation, $t_L = 25$ seconds and $t_H= 10$ seconds, described in Section~\ref{sec:resourceAllocationProblem}. 
    \item  \textbf{Confidence-based time with explanation}: This is our second treatment group, where in addition to confidence-based time allocation, we provide the AI confidence (``low'' or ``high'') corresponding to each trial. 
\end{enumerate}
Out of the 479 participants, $95$ were in \ho, $109$ were in \const{}, $95$ were in \random{}, $85$ were in \conf{} and $96$ were in \confExp{}. In the groups where participants switch between $10$-second and $25$-second conditions the accuracy would likely be affected by rapid switching between the two time conditions. Hence, we created blocks of 5 trials with the same time allocation for both groups. The complete testing section contained 8 such blocks. This concludes the description of Experiment 2.
\subsection{Results} 
\label{sec:results2}

\begin{figure}%
    \centering
    \subfloat[Average accuracy]{\includegraphics[width=12.2cm]{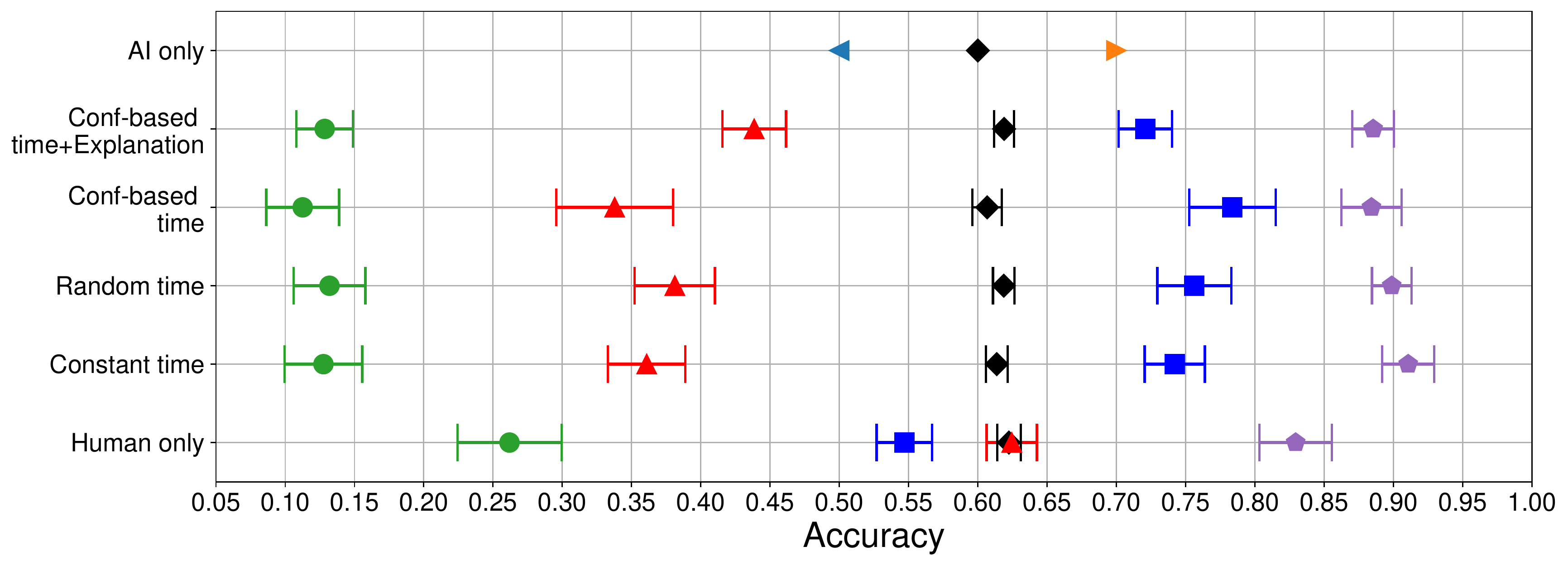} }\\%
    \subfloat[Average agreement with AI]{\includegraphics[width=12.2cm]{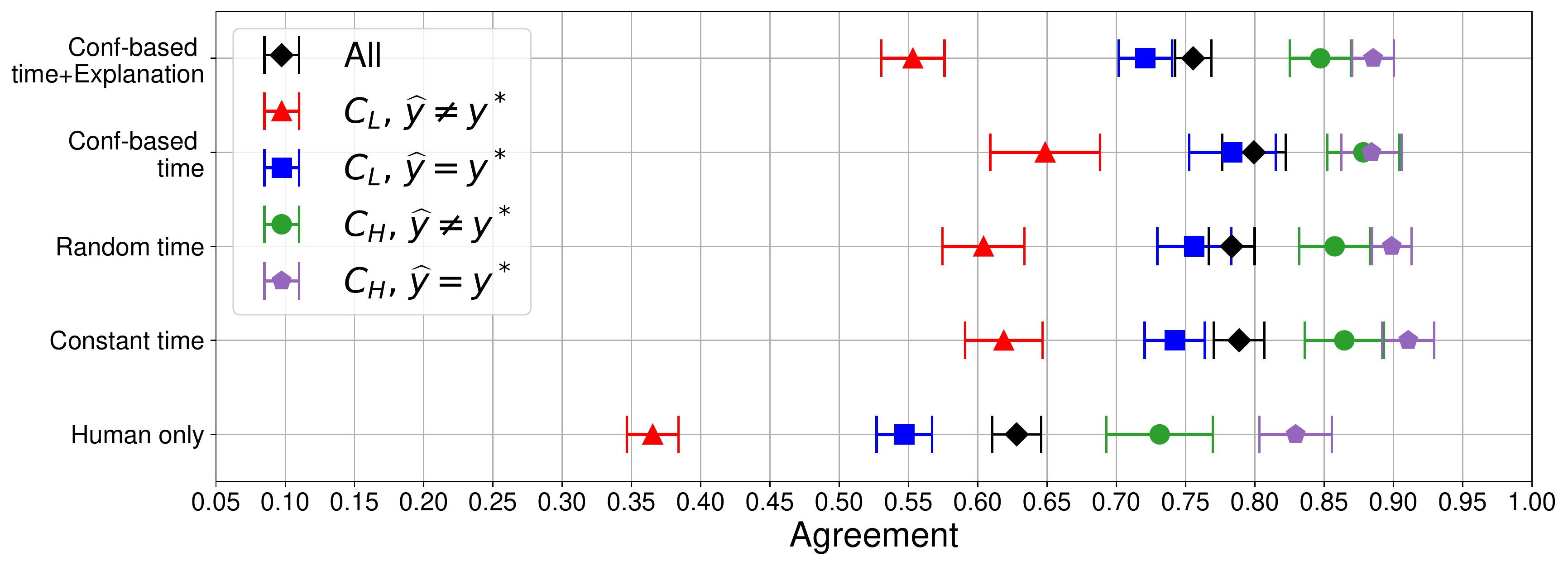} }
    \caption{Average accuracy and agreement ratio of participants in Experiment 2 across the four different conditions, marked on the y-axis. We note that the error bars in (a) for 'All' trials (black diamonds) are smaller than the marker size. }
    \label{fig:AIconfDist}%
\end{figure}
Figure~\ref{fig:AIconfDist} shows that our effort to create a scenario where the AI knowledge is complementary to human knowledge is successful because the AI only and \ho{} conditions have similar overall accuracy (around $60\%$, black diamonds), and yet humans only agreed with the AI in $62.3\%$ of the trials. Moreover, on trials where AI is incorrect, \ho{} has accuracy of $61.8\%$ on trials in $\Cset_L$, and $29.8\%$ on trials in $\Cset_H$. Thus, the participants showed more complementary knowledge in trials in $\Cset_L$ compared to $\Cset_H$.

Given this successful setup of complementary knowledge between humans and AI, there is good potential for the human-AI partnership groups, especially the \conf{} group and the \confExp{} group, to outperform the AI only or \ho{} groups (H4). In Figure~\ref{fig:AIconfDist}(a), we see that the mean accuracy of the human-AI team is $61\%$ in \conf{} and $61.9\%$ in \confExp{} while the accuracy of \ho{} is $61.9\%$ and the accuracy of the AI model is $60\%$. Thus, regarding H4, the results suggest that the accuracy in \conf{} is greater than AI alone $(p=0.06,t(183) = 1.52)$, whereas they do not provide sufficient evidence for \conf{} being better than \ho{} $(p=0.58,t(92) = -0.21)$. Similarly, regarding H6, the results suggest that the accuracy in \confExp{} is better than AI alone $(p=0.004,t(194) = 2.66)$, whereas for \ho{} the results are not statistically significant $(p=0.5,t(189) = -0.02)$.

However, we see that anchoring bias affected overall team performance negatively when the AI is incorrect (H2). Figure~\ref{fig:AIconfDist}(b) shows evidence of anchoring, the agreement percentage in the \ho{} group is much lower than those in the collaborative conditions $(p<0.001, t(184)=6.73)$. When the AI was incorrect (red triangles and green circles), this anchoring bias clearly reduced team accuracy when compared to the \ho{} accuracy $(p<0.001, t(370)=-6.68)$. Although, it is important to note that the \ho{} group received longer time ($25$s) than the collaborative conditions on average. Nevertheless, if we just compare \ho{} and \conf{} within the low confidence trials (red triangles), where both were assigned the same amount of time($25$s), we observe similar disparity in agreement percentages $(p<0.001, t(92)=4.97)$ and accuracy $(p<0.001, t(92)=-4.74)$. Hence, the results are consistent with H2.

Regarding H3, we see that while \conf{} alone did not lead to sufficient adjustment away from the AI when it was incorrect, \confExp{} showed significant reduction in anchoring bias in the low confidence trials (red triangles) compared to the other conditions  $(p=0.003 , t(383)= 2.70)$, which suggests that giving people more time along with an explanation for the time helped them adjust away from the anchor sufficiently, in these trials (H3). This de-anchoring also led to higher accuracy in these trials(red triangles) for \confExp{} ($43.8\%$) when compared to the other three collaborative conditions with \random{} at $36.2\%$, \const at $36.4\%$ and \conf{} at $37.5\%$. Note that the set of conditions chosen in our experiment does not allow us to separately quantify the effect of the time-based allocation strategy and the confidence-based explanation; we discuss this in Section~\ref{sec:discussion}.

Next, we examine the differences between the four collaborative groups. Figure~\ref{fig:AIconfDist}(a) shows that the average accuracy over all trials (black diamonds) is highest for \confExp{} at $61.9\%$ with \conf{} at $61\%$, \random{} at $61.1\%$ and \const{} at $61.5\%$. Regarding H5, we see that \conf{} does not have significantly different accuracy from the other collaborative groups. Finally, regarding H6, we observe that \confExp{} has the highest accuracy, although the effect is not statistically significant $(p=0.19, t(383)=0.84)$. We note that the outcomes in all collaborative conditions are similar in all trials except trials where AI is incorrect and has low confidence, and in these trials our treatment group has significantly higher accuracy. This implies that in settings prone to over-reliance on AI, \confExp{} helps improve human-AI team performance. 

The reason that the overall accuracy of \conf{} is not significantly better than the other two collaborative conditions is likely because of the relatively low accuracy and low agreement percentage in trials in $\Cset_H$ (green circles, purple pentagons). Based on the results of Experiment 1, we expected that the agreement percentage for the $10$-second trials would be high and since these align with the high AI confidence trials for \conf{}, we expected these trials to have a high agreement percentage and hence high accuracy. Instead, we observed that \conf{} has low agreement percentage $(84\%)$ in $\Cset_H$, compared to \random{} $(87.9\%)$, and \const{}$(88.1\%)$, both having an average time allocation of $17.5$ seconds. This lower agreement percentage translates into lower accuracy ($86\%$) when AI is correct (purple pentagons). In the next section, we discuss how this points to possible distrust of AI in these high confidence trials and its implications. For \confExp{} we observe that the participants in this group are able to de-anchor from incorrect low confidence AI predictions, to give higher mean accuracy than other collaborative conditions, albeit the difference is not statistically significant. 

\section{Discussion}
\label{sec:discussion}
\paragraph{\textbf{Lessons learned.}} We now discuss some of the lessons learned from the results obtained in Experiment 2. As noted in Section~\ref{sec:results2}, we see that \conf{} has a low agreement rate on trials in $\Cset_H$ where the time allocated is $10$ seconds and the AI prediction is $70\%$ accurate. Moreover, we see that the agreement rate is lower than \ho{} and \const{} on trials in $\Cset_H$, where the AI prediction is correct as well as where the AI prediction is incorrect. This behavior suggests that the participants in \conf{} may have grown to distrust the AI, as they disagreed more with the AI on average and spent more time on the trials where they disagreed. The distrust may be due to \conf{} assigning longer times ($25s$) only to low-AI-confidence trials, perhaps giving the impression that the AI is worse than it really is. However, these effects are reduced by providing an explanation for the time allocation in \confExp{}. Our observations highlight the importance of accounting for human behaviour in such collaborative decision-making tasks. 

Another insight gained from Experiment 2 is that the model should take into account the sequentiality of decision-making where the decision-maker continues to learn and build their perception of the AI as the task progresses, based on their interaction with the AI. Dynamic Markov models have been studied previously in the context of human decision-making~\citep{busemeyer2020comparison, Lieder2018anchoringBias}. We believe that studying dynamic cognitive models that are cognizant of the changing interaction between the human and the AI model would help create more informed policies for human-AI collaboration.

\paragraph{\textbf{Limitations.}} One limitation of our study is that our participants are not experts in student assessment. To mitigate this problem we first trained the participants in the task and showed them the statistics of the problem domain. We also showed more features to the human users, compared to the AI, to give them complementary knowledge. The fact that human-only accuracy in Experiment 2 is roughly the same as the AI-only accuracy suggests that these domain-knowledge enhancement measures were effective. Secondly, we proposed a time-based strategy and conducted Experiment 1 to validate our hypothesis and select the appropriate time durations $(10s, 25s)$ for our second experiment. Due to limited resources, we did not extend our search space beyond four settings -- $(10s, 15s, 20s, 25s)$. Although it is desirable to conduct the experiment with real experts, this can be extremely expensive. Our approach can be considered as "human grounded evaluation"~\cite{doshi2017towards}, a valid approach by using lay people as a "proxy" to understand the general behavioral patterns. We used a non-critical decision-making task where the participants would not be held responsible for the consequences of their decisions. This problem was mitigated by introducing an outcome-based bonus reward which motivates optimal decision-making. Our work considers the effect of our time allocation strategy with and without the confidence-based explanation through the treatment groups in experiment 2. While this helps us investigate the benefits of the time allocation strategy, we cannot separate out the independent effect of the confidence-based explanation strategy. Lastly, our work focuses on a single decision-making task. Additional work is needed to examine if the effects we observe generalize across domains and settings. However, prior research provides ample evidence that even experts making critical decisions resort to heuristic thinking, which suggests that our results will generalize broadly.

\paragraph{\textbf{Conclusions.}} In this work, we foreground the role of cognitive biases
in the human-AI collaborative decision-making setting. Through literature in cognitive science and psychology, we explore several biases and present mathematical models of their effect on collaborative decision-making. We focus on anchoring bias and the associated anchoring-and-adjustment heuristic that is important towards optimizing team performance. We validate the use of time as an effective strategy for mitigating anchoring bias through a user study. Furthermore, through a time-based resource allocation formulation, we provide an optimal allocation strategy that attempts to achieve the "best of both worlds" by capitalizing on the complementary knowledge presented by the decision-maker and the AI model. Using this strategy, we obtain human-AI team performance that is better than the AI alone, as well as better than having only the human decide in cases where the AI predicts correctly. When the AI is incorrect, the information it provides the human distracts them from the correct decision, thus reducing their performance. Giving them information about the AI confidence as explanation for the time allocation alleviates some of these issues and brings us closer to the ideal Human-AI team performance shown in Figure \ref{fig:idealhumanAI}.

\paragraph{\textbf{Future work.}} Our work shows that a time-based strategy with explanation, built on the cognitive tendencies of the decision-maker in a collaborative setting, can help decision-makers adjust their decisions correctly. More generally, our work showcases the importance of accounting for cognitive biases in decision-making, where in the future we would want to study other important biases such as confirmation bias or weak evidence effect. This paper opens up several directions for future work where explanation strategies in this collaborative setting are studied and designed based on the cognitive biases of the human decision-maker. Another interesting direction is to utilize the resource allocation framework for other cognitive biases based on their de-biasing strategies. \\

\noindent \textbf{Acknowledgements.} The work of Charvi Rastogi was supported in part by NSF CIF 1763734.

\bibliographystyle{apalike}
\bibliography{bibtex}

~\\~\\
\appendix

~\\\noindent{\bf \Large Appendix}

\section{Additional details of optimal resource allocation}

Following from the discussion in Section~\ref{sec:resourceAllocationProblem}, in this section we provide additional details about Assumption~1 and the optimality of the confidence-based time allocation policy proposed thereafter.

\paragraph{Reasoning for Assumption 1.} To see how Assumption 1~\eqref{eq:assumption1} might hold, we refer first to Figure~\ref{fig:idealhumanAI}, which assumes that human accuracy is higher than AI accuracy when confidence is low, $\Cest_i \in \Cset_L$, and lower when $\Cest_i \in \Cset_H$. (Human accuracy does not have to be uniformly higher/lower in $\Cset_L$/$\Cset_H$ as Figure~\ref{fig:idealhumanAI} suggests.) At $t = 0$, $\E[R_i \given \Cest_i \in \Cset, \tottime{i}  = 0]$ is equal to AI accuracy conditioned on $\Cset = \Cset_L, \Cset_H$, and by giving the human more time to de-anchor, we might expect $\E[R_i \given \Cest_i \in \Cset_L, \tottime{i}  = t]$ to increase and $\E[R_i \given \Cest_i \in \Cset_H, \tottime{i}  = t]$ to decrease. A second way to understand Assumption 1 is to break down the conditional accuracy into two parts:  
\begin{align}
\nonumber \E[R_i \given \Cest_i \in \Cset, \tottime{i}  = t] =\,\,  & \P(\perceived_i = \predict_i \given \predict_i= \truth_i, \Cest_i \in \Cset, \tottime{i}  = t)\P(\predict_i= \truth_i \given \Cest_i \in \Cset)  \\&\quad +\P(\perceived_i \neq \predict_i \given \predict_i \neq \truth_i, \Cest_i \in \Cset, \tottime{i}  = t)\P(\predict_i \neq \truth_i \given \Cest_i \in \Cset), 
\label{eq:fourPart}
\end{align}
for $\Cset = \Cset_L, \Cset_H$. The results of Section~\ref{sec:prelimStudy} indicate that the disagreement probability in the second RHS term in \eqref{eq:fourPart} increases or stays the same with time $t$, and the agreement probability in the first term decreases or stays the same with $t$. For $\Cset = \Cset_L$, assuming positive correlation between low confidence $\Cest_i \in \Cset_L$ and low accuracy (not necessarily the perfect correlation in Figure~\ref{fig:idealhumanAI}), the probability $\P(\predict_i \neq \truth_i \given \Cest_i \in \Cset_L)$ tends to be larger and the second term dominates, resulting in the LHS increasing with $t$. For $\Cset = \Cset_H$, the first term tends to dominate, leading to decrease with $t$. \\

To show why the confidence-based policy has higher accuracy, we provide the following corollary.

\begin{corollary}\label{cor:constRand}
Consider the two time allocation strategies defined above - (1) Constant time and (2) Random time, in AI-assisted decision-making where we have total $\num$ trials and time $\tottime{}$. Suppose Assumption 1 stated in~\eqref{eq:assumption1} holds. Then the human-AI accuracy of the confidence-based time allocation policy is greater than or equal to the accuracy of the constant time allocation and random time allocation strategies.
\end{corollary}

\begin{proof}
Let the two-level confidence based allocation policy be denoted by $\pi$. Now, we have that the accuracy for each round under this policy, 
\begin{align}
    \E_\pi[R_i] &= \E\left[ \E[R_i \given \Cest_i, \tottime{i}=\pi(\Cest_i)] \right]\nonumber\\
    &= \P(\Cest_i \in \Cset_L) \E[R_i \given \Cest_i \in \Cset_L, \tottime{i} = t_L] + \P(\Cest_i \in \Cset_H) \E[R_i \given \Cest_i \in \Cset_H, \tottime{i} = t_H].
\end{align}
For the constant allocation policy, we get
\begin{align}
    \E_{\mathrm{const}}[R_i] &= \E[R_i \given \tottime{i} = \frac{\tottime{}}{ \num}]\nonumber\\
    &= \P(\Cest_i \in \Cset_L) \E[R_i \given \Cest_i \in \Cset_L, \tottime{i} = \frac{\tottime{}}{ \num}] + \P(\Cest_i \in \Cset_H) \E[R_i \given \Cest_i \in \Cset_H, \tottime{i} = \frac{\tottime{}}{ \num}].
\end{align}
Now, according to assumption 1 in~\eqref{eq:assumption1}, we have $ \E[R_i \given \Cest_i \in \Cset_L, \tottime{i} = t_L]\geq \E[R_i \given \Cest_i \in \Cset_L, \tottime{i} = \frac{\tottime{}}{ \num}]\nonumber$, and $ \E[R_i \given \Cest_i \in \Cset_H, \tottime{i} = t_H] \geq \E[R_i \given \Cest_i \in \Cset_H, \tottime{i} = \frac{\tottime{}}{ \num}]$. Thus, $E_\pi[R_i] \geq \E_{\mathrm{const}}[R_i]$. 
Similarly for random allocation, we have
\begin{align}
    \E_{\mathrm{rand}}[R_i] &= \P( \tottime{i}= t_L) \E[R_i \given \tottime{i} = t_L] + \P(\tottime{i} = t_H) \E[R_i \given \tottime{i} = t_H] \nonumber\\
    &= \P(\Cest_i \in \Cset_L) \P(\tottime{i} = t_L) \E[R_i \given \Cest_i \in \Cset_L, \tottime{i} = t_L]  \nonumber\\
    &\quad {} + \P(\Cest_i \in \Cset_H) \P(\tottime{i}= t_L) \E[R_i \given \Cest_i \in \Cset_H, \tottime{i} = t_L]\nonumber\\
    &\quad {} + \P(\Cest_i \in \Cset_L) \P(\tottime{i} = t_H) \E[R_i \given \Cest_i \in \Cset_L, \tottime{i} = t_H]  \nonumber\\
    &\quad {} + \P(\Cest_i \in \Cset_H) \P(\tottime{i} = t_H) \E[R_i \given \Cest_i \in \Cset_H, \tottime{i} = t_H].
\end{align}
Using the assumptions~\eqref{eq:assumption1} as stated for constant allocation, we prove that $E_\pi[R_i] \geq \E_{\mathrm{rand}}[R_i]$. 
\end{proof}

\end{document}